\documentclass[a4paper,UKenglish]{article}
\usepackage{amsthm}
\usepackage[margin=2cm]{geometry}

\usepackage{wrapfig}
\usepackage[utf8]{inputenc}
\usepackage{booktabs}
\usepackage{subcaption}
\usepackage{graphicx}
\usepackage{txfonts}

\usepackage{hyperref}

\newtheorem{lemma}{Lemma}

\usepackage{microtype}

\title{Computing Tree Decompositions with FlowCutter: PACE 2017 Submission\footnote{Support by DFG grant WA654/19-2 ``Algorithm Engineering für Graphpartitionierung''}}

\author{Ben Strasser\\
Karlsruhe Institute of Technology (KIT), Karlsruhe, Germany\\
  \texttt{strasser@kit.edu}}

\usepackage{babel}

\begin{document}

\maketitle

\begin{abstract}
We describe the algorithm behind our PACE 2017 submission to the heuristic tree decomposition computation track.
It was the only competitor to solve all instances and won a tight second place.
The algorithm was originally developed in the context of accelerating shortest path computation on road graphs using multilevel partitions.
We illustrate how this seemingly unrelated field fits into tree decomposition and parameterized complexity theory.
 \end{abstract}

\section{Introduction}

Tree decompositions are an established graph decomposition methodology.
They can be used to measure how ``close'' a graph is to a tree.
The width of a decomposition indicates how ``tree-like'' a graph is.
We refer to~\cite{bp-aicgc-93},~\cite{b-atgt-93}, and~\cite{b-tsa-07} for an overview of the field.

Tree decompositions are often used in combination with parameterized complexity.
Many NP-hard problems can be solved in linear time on tree graphs.
The corresponding linear time algorithms can often be generalized to tree decompositions.
The generalized algorithms' running times are usually linear in the graph size but super-polynomial in the width of the decomposition.

A lot of theoretical research with in-depth results into this field exists.
However, it is unclear whether these results translate into algorithms that are efficient in practice.
A core issue consists of finding decompositions of small size.
Without an algorithm that is fast in practice and computes tree decompositions with a reasonably small width, no algorithm parameterized in the tree width is usable in practice.

To bridge this gap between theory and practice, the PACE implementation challenge was created in 2016~\cite{dhjkkr-tfpac-16}.
Track A2 focuses on computing tree decompositions of small width within a given amount of time.
Because of its practical relevance, the challenge was repeated in 2017.
In this paper, we describe the FlowCutter submission to the PACE 2017 Track A2 contest.
It was the only submission so solve all test instances and won a close second place.

FlowCutter was originally developed to accelerate shortest path computations on road graphs.
Fortunately, it is applicable in a significantly broader context.
The corresponding PACE 2016 and 2017 submissions demonstrate this.
For a survey of shortest path computation algorithms, we refer to~\cite{bdgmpsww-rptn-16}.
Tree decompositions are rarely used explicitly in this community.
However, very similar concepts are used.

Many shortest path acceleration techniques exist.
They work in two phases: 
In the preprocessing phase, the graph is transformed.
In the query phase, shortest paths are computed using the transformed graph.
The preprocessing phase may be slow.
The query phase should be fast.
The motivation is that the preprocessing only needs to run when the map is updated.
This is assumed to be rare.
For example, in the preprocessing phase one can construct a small tree decomposition.
In the query phase, the decomposition can be used to compute shortest paths.

These shortest path computation algorithms are an example of parameterized complexity.
The algorithms assume that a small tree decomposition is given in the input.
The complexity of the shortest path problem is bounded in terms of the tree width.
The difference to regular fixed parameter tractability theory is that the considered problem is not NP-hard.

Two common shortest path techniques are Contraction Hierarchies (CH) \cite{gssv-erlrn-12} and Multilevel Dijkstra (MLD)~\cite{sww-daola-00,hsw-emlog-08,dgpw-crprn-13}.
None of the original publications mention tree decompositions.
However, very similar concepts are used.
The similarity is a comparatively recent discovery.
We illustrate the relationship in this paper.

A question often asked is what algorithms big companies with mapping services use.
Unfortunately, most companies do not publicly share this information.
A welcome exception is Microsoft.
They stated that the Bing routing service uses MLD~\cite{bingcrp}.
This illustrates that research in this domain has direct practical implications.

The relation between multilevel graph partitions and tree decompositions is not well known as none of the cited survey articles mention it~\cite{bp-aicgc-93,b-atgt-93,b-tsa-07,bdgmpsww-rptn-16}.
The relationship was hinted in~\cite{acdgw-o-16}.
In this paper, we make the connection clearer and more explicit.
There is a one-to-one correspondence between multilevel graph partitions and rooted tree decompositions.
Using this one-to-one relationship, we can also establish a link between CH and MLD via tree decomposition theory.
We believe that this connection is useful also in other contexts than shortest path computations.
It allows to interpret tree decompositions as multilevel graph partitions.

Our PACE 2017 submission uses the multilevel graph view.
It is based upon recursive bisection respectively nested dissection.
It uses FlowCutter~\cite{hs-gbpo-16} to bisect graphs.

\subsection{Outline}

We first define the terminology used in this paper.
Afterwards, we present a high level overview of related shortest path algorithms.
We then present in the next step, how multilevel partitions and tree decompositions relate.
Using the multilevel partition view, we describe our PACE 2017 algorithm.
Finally, we present the PACE 2017 results.

\section{Definitions}

In this section, we start by defining standard tree decomposition and chordal graph concepts. 
Afterwards, we formally define the notion of multilevel partition.

An \emph{edge cut} $C$ of a graph $G=(V,E)$ is a non-empty edge set.
$C$ decomposes $G$ into the connected components of $G\setminus C$.
Similarly, a \emph{node separator} $S$ of a graph $G=(V,E)$ is a non-empty node set. 
$S$ decomposes $G$ into the connected components of $G\setminus S$.
Usually, one requires there to be at least two connected components in $G\setminus S$.
However, we allow for the degenerate case of there only being one connected component.
We say that a cut or separator \emph{separates} two nodes, if they are in different components. 

An undirected graph is \emph{chordal}, if for every cycle $Z$ with at least four nodes, there exists an edge between two nodes of $Z$ that are not adjacent within $Z$.
A \emph{chordal supergraph} $G'$ of a graph $G$ is a supergraph of $G$ that is chordal.
\emph{Triangulated graph} is a synonym for chordal graph, also used in the literature.
An \emph{elimination order} of an undirected graph $G$ is an order $O$ of the nodes of $G$. 
A supergraph $G'$ is obtained by iteratively contracting the nodes of $O$, i.e., removing a node and adding a clique among its neighbors.
In the shortest path literature, the term \emph{contraction order} is used.
A \emph{perfect elimination order} of a graph $G$ is an order where no node contraction inserts edges.
A graph has a such an order, if and only if, it is chordal~\cite{fg-imig-65}.
We therefore refer to $G'$ as \emph{chordal supergraph} of $G$.

A \emph{tree decomposition} of a graph $G=(V,E)$ is a pair $(B,T)$, where $B$ is the set of \emph{bags} and $T$ is the \emph{tree backbone}.
Every bag $b\in B$ is a set of nodes, i.e., $b\subseteq V$. 
$T$ is a tree where the bags are the nodes, i.e., $B$ is the set of nodes of $T$.
A tree decomposition must fulfill three criteria to be valid:
\begin{enumerate}
\item Every node is in a bag, i.e., $\bigcup_{b\in B} b = V$.
\item For every edge $\{x,y\}$ of $G$, there must be a bag $b\in B$ such that both end points are in $b$, i.e., $x\in b$ and $y \in b$.
\item For every node $x$, the subgraph of the tree backbone $T$ induced by all bags that contain $x$ is a tree.
\end{enumerate}
A \emph{rooted tree decomposition} is a tree decomposition, where the tree backbone is directed towards a root bag~$r$.
A rooted tree decompositions is illustrated in Figure~\ref{fig:rooted-td}.
The \emph{width} of a tree decomposition is the maximum size of a bag plus one.
The tree width $\textrm{tw}$ of a graph is the minimum width over all tree decompositions.

\begin{figure}
\begin{center}
\begin{subfigure}[b]{.3\linewidth}
\includegraphics{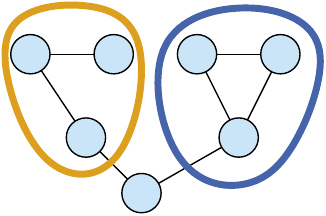}
\caption{No touch}
\end{subfigure}
~~
\begin{subfigure}[b]{.3\linewidth}
\includegraphics{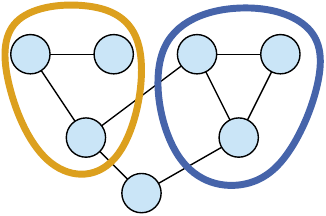}
\caption{Touch}
\end{subfigure}
~~
\begin{subfigure}[b]{.3\linewidth}
\includegraphics{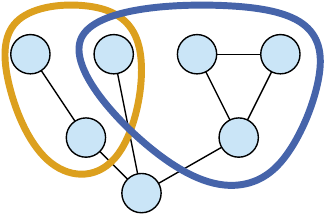}
\caption{Also Touch}
\end{subfigure}
\end{center}
\caption{Graph with two cells illustrating the touch definition.}
\label{td:fig:touch}
\end{figure}

A \emph{cell} $c$ of a graph $G=(V,E)$ is a node subset $c\subseteq V$.
Two cells $a$ and $b$ \emph{touch}, if there is an edge with endpoints in $a$ and $b$ or when $a$ and $b$ share a node.
The touching relationship is illustrated in Figure~\ref{td:fig:touch}.
The \emph{boundary} $b$ of $c$ is the set of nodes adjacent to a node in $c$ but not in $c$.

A \emph{multilevel partition} $P$ is a set of cells.
Two criteria must be fulfilled for $P$ to be valid:
\begin{enumerate}
\item $V$ is a cell. We refer to $V$ as the \emph{toplevel cell}.
\item Touching implies nesting. If two cells $a$ and $b$ touch then $a\subseteq b$ or $b\subseteq a$.
\end{enumerate}
$p$ is a \emph{parent cell} of $c$ if $c\subset p$ and no other cell $q$ exists with $c\subset q \subset p$.
Analogous, $c$ is a child of $p$.
In a \emph{bilevel partition}, only the toplevel cell has children.
A bilevel partition is illustrated in Figure~\ref{fig:bilevel} and a multilevel partition in Figure~\ref{fig:multilevel}.

The term ``multilevel partition'' is often used in experimental algorithms papers \cite{bmsss-ragp-13} but rarely formally defined.
It usually describes an informal algorithm design pattern. 
The exact details therefore significantly vary between application-focused papers.
However, there is a significant difference between our formalization and many other papers.
We separate cells using node separators.
Many application-focused papers use edge cuts.
We use separators as it enables a tighter coupling with tree decompositions.

\section{Application: Shortest paths in Road graphs}

Multilevel Dijkstra (MLD) and Contraction Hierarchies (CH) are two very successful shortest path acceleration techniques.
In this section, we briefly outline the algorithms.

\subsection{Multilevel Dijkstra}

\begin{figure}
\begin{center}
\begin{subfigure}[c]{.3\linewidth}
\includegraphics[width=\linewidth,page=1]{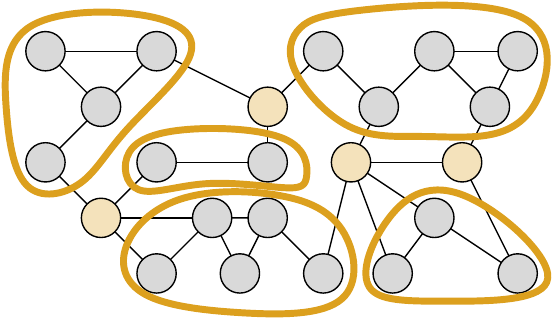}
\caption{Input graph with partition\\(top level cell not depicted)}\label{fig:bilevel}
\end{subfigure}~%
\begin{subfigure}[c]{.3\linewidth}
\includegraphics[width=\linewidth,page=2]{mld}
\caption{Boundary nodes with overlay cliques}\label{fig:overlay}
\end{subfigure}~%
\begin{subfigure}[c]{.3\linewidth}
\includegraphics[width=\linewidth,page=3]{mld}
\caption{Graph search for query from node $s$ to node $t$}\label{fig:mld-query}
\end{subfigure}
\end{center}
\caption{MLD example on bilevel partition.}
\end{figure}

The core idea of Multilevel Dijkstra (MLD) was introduced in~\cite{sww-daola-00} using a bilevel partition.
It was extended to multiple levels in~\cite{hsw-emlog-08}.
The algorithm was refined and popularized by~\cite{dgpw-crprn-13}.
A parameterized analysis of the algorithm using tree width is provided by~$\cite{acdgw-o-16}$.
We describe the initial bilevel algorithm as it is the simplest.
This technique is also often called multilevel overlay graph (MLO) or CRP which references the title of~\cite{dgpw-crprn-13}.

The first step of the preprocessing consists of computing a set of non-touching cells.
This is illustrated in Figure~\ref{fig:bilevel}.
In the second step, an \emph{overlay} is computed for every cell $c$. 
Denote by $G_c$ the subgraph induced by $c$ and its boundary.
Further, denote by $H_c$ the overlay graph.
$H_c$ is substitution for $G_c$ that maintains the shortest path distances among the boundary nodes.
$H_c$ must contain at least the boundary nodes.
In theory, $H_c$ can be very complex graphs.
However, usually a clique among the boundary nodes is used.
The weights of the clique edges are computed by running Dijkstra's algorithm from every boundary node restricted to $G_c$.
Figure~\ref{fig:overlay} illustrates the overlays.
The query consists of running a bidirectional variant of Dijkstra's algorithm.
It explores the cells of $s$ and $t$ and replaces all other cells with their overlays.
Figure~\ref{fig:mld-query} illustrates the explored subgraph.

The idea of this algorithm can easily be applied recursively.
The extension uses a multilevel partition as input. 
In $\cite{acdgw-o-16}$, the MLD distance query running time with clique overlays was bounded by $O(\textrm{tw}^2 (\log n) (\log(\textrm{tw} \log n)))$.

\subsection{Contraction Hierarchies}

The algorithm was originally described in~\cite{gssv-erlrn-12}.
A newer variant of the algorithm is called Customizable Contraction Hierarchies (CCH)~\cite{dsw-cch-15}.
We outline the later variant, because it is tighter coupled with tree decompositions.

The preprocessing consists of several steps.
In the first step, a contraction order $O$ is computed using nested dissection.
In the second step, the nodes are contracted in this order.
The inserted edges are called \emph{shortcuts}.
The graph plus the shortcuts is the CCH.
In different contexts, the contraction order is called elimination order, the shortcuts are the \emph{fill-in}, and the CCH is the chordal supergraph.

The order $O$ is usually interpreted as ordering the nodes from bottom to top.
The top node is the last node of $O$ and contracted last.
A node $x$ is \emph{lower} than $y$, if $x$ comes before $y$ in $O$.
Denote by $w(x,y)$ the weight of edge $\{x,y\}$.
The objective is to enforce the \emph{lower triangle inequality}.
For every triangle $\{x,y,z\}$ in the CCH, it requires that $w(x,z)+w(z,y) \ge w(x,y)$, where $z$ is lower than $x$ and $y$.
The lower triangle inequality is enforced using a triangle listing algorithm.
All edges already present in the input graph are assigned their input weight.
Shortcuts are assigned $\infty$.
Triangles $\{x,y,z\}$ are enumerated ordered by the position of the lowest node $z$ in $O$.
If $w(x,z)+w(z,y)< w(x,y)$, then $w(x,y)$ is set to $w(x,z)+w(z,y)$.

A query for a $st$-path consists of a bidirectional graph search from $s$ and $t$.
From a node $x$ both searches only follow edges to nodes higher than $x$.
Once both searches are terminated, a shortest \emph{up-down} $st$-path $P$ is found.
$P$ is a path $v_1{\rightsquigarrow} v_i {\rightsquigarrow} v_k$ such that the nodes between $v_1$ and $v_i$ are ascending according to $O$.
The nodes from $v_i$ to $v_k$ are descending according to $O$.
The query is correct, if there always exists an up-down path that is as long as a shortest path.
To show this consider a shortest path $P$. 
If $P$ is an up-down path, we are done.
Otherwise, there exists a subpath $x{\rightarrow} z {\rightarrow} y$ such that $z$ is lower than $x$ and $y$.
As $z$ is contracted before $x$ and $y$, there exists a shortcut from $x$ to $y$.
Further, because of the lower triangle inequality, the shortcut's weight is as long as the path $x{\rightarrow} z {\rightarrow} y$.
We can thus remove $z$ from $P$.
Either $P$ is now an up-down path or we apply the argument iteratively.
As we eventually always end up with a shortest up-down path, the query is correct.

The maximum cliques in a CCH/chordal supergraph correspond to the tree decomposition.
It was shown in~\cite{bcrw-s-16} that the number of nodes visited by one side is bounded by $O(\textrm{tw} \log n)$.
As the explored subgraphs can be dense, the distance query runs in $O(\textrm{tw}^2 \log^2 n)$ time.

\section{Connection between Tree Decompositions and Multilevel Partitions}

In this section, we prove a one-to-one correspondence between multilevel partitions and rooted tree decompositions.
A consequence is that both concepts are just two different views onto the same object.

We first describe the algorithm to translate a multilevel partition into a rooted tree decomposition.
In the next step, we describe the algorithm to perform the inverse operation.
Finally, we show that the outputs of both algorithms are valid and that chaining both algorithms is the identity function.

\subsection{From Multilevel Partition to Rooted Tree Decomposition.}

For every cell $c$ in the multilevel partition, we construct a bag $b$ in a tree decomposition.
$b$ is the union of the boundary and interior nodes of $c$ minus the interior nodes of all children.
The parent-child relation between cells induces a tree on the bags. 
This is the tree backbone.
The top level cell is the root of the rooted tree backbone.
The so obtained tree decomposition can be degenerate, i.e., it is possible that bags exist that are subsets of other bags. 

\subsection{From Tree Decomposition to Multilevel Partition}

A tree decomposition $T$ does not uniquely define a multilevel partition $P$ because the tree backbone does not have a root.
The transformation from $T$ must therefore start by picking a root bag $r$.
With respect to $r$, we can construct for every bag $b$ except the root a cell as follows:
Denote by $p$ the parent of $b$, i.e., the first node on the unique path from $b$ to $r$. 
We set the boundary of cell $c$ to $b\cap p$.
The interior of $c$ is set to the union of all direct or indirect children bags of $c$ minus $c$'s boundary.
We additionally construct a cell for the root bag.
This cell's boundary is empty and its interior is the whole graph.

\subsection{Example}

\begin{table}
\begin{center}
\begin{tabular}{ccc}
\toprule
Cell Interior & Cell Boundary & Corresponding Bag \\
\midrule 
a                                                    & b, f        & a, b, f\\
a, b                                                 & f, c        & b, f, c\\
n                                                    & m, o        & n, o, m\\
n, o                                                 & m, p        & o, m, p\\
m, o, n                                              & f, p        & f, m, p\\
g                                                    & f, c, p     & f, g, p, c\\
a, b, f, g, m, n, o                                  & p, c        & f, p, c\\
j, k, l                                              & r, i        & j, r, i\\
k, l                                                 & j           & k, l, j\\
q                                                    & p, i, r     & q, r, p, i\\
d                                                    & c, e        & d, e, c\\
h                                                    & c, e, p     & h, e, p, c\\
d, e, h                                              & c, p, i     & e, i, p, c\\
q, r, j, k, l                                        & p, i        & p, r, i\\
d, e, h, i, q, r, j, k, l                            & p, c        & p, c, i\\
a, b, c, d, e, f, g, h, i, j, k, l, m, n, o, p, q, r & $\emptyset$ & p, c\\
\bottomrule
\end{tabular}
\end{center}
\caption{Interior, boundary, and bags of the cells of the multilevel partition of Figure~\ref{td:fig:multilevel}.}
\label{td:tag:multilevel} 
\end{table}

\begin{figure}
\begin{center}
\begin{subfigure}[c]{.45\linewidth}
\includegraphics[width=\linewidth]{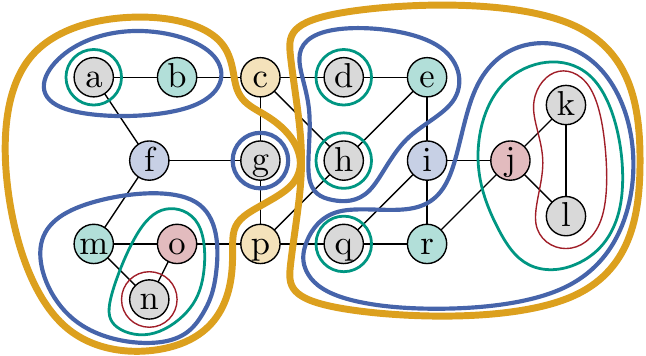}
\caption{Multilevel Partition \\ (top level cell not depicted)}\label{fig:multilevel}
\end{subfigure}
~
\begin{subfigure}[c]{.45\linewidth}
\includegraphics[width=\linewidth]{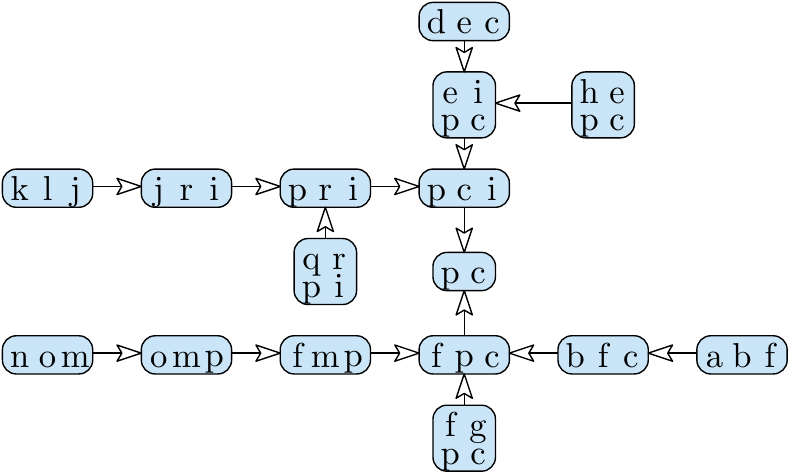}
\caption{Rooted Tree Decomposition}\label{fig:rooted-td}
\end{subfigure}
\end{center}
\caption{Multilevel Partition and corresponding rooted tree decomposition.}
\label{td:fig:multilevel}
\end{figure}

Figure~\ref{fig:multilevel} depicts a multilevel partition with 16 cells.
Every cell (except the top level cell) is depicted as closed curve. 
The color of a curve indicates the recursion depth. Orange indicates depth 1, blue depth 2, green has depth 3, and red indicates depth 4.
Grey nodes are not part of any separator. The color of the remaining nodes indicates the depth of the separator that they are part of.
Table~\ref{td:tag:multilevel} enumerates all cells in the multilevel partition and the derived bags.
Together with the parent-child relation of the cells, we obtain the rooted tree decomposition depicted in Figure~\ref{fig:rooted-td}.

There are cells, such as the cell with interior $\{a,b\}$, which is ``divided'' along the separator $\{b\}$ into one part, namely the cell with interior $\{a\}$.
Sufficiently large cells are usually divided into more than one part. 
However, for tiny cells, it often occurs that such awkward separators are used.

\begin{wrapfigure}[15]{o}{7cm}

\begin{center}
\vspace{-2em}
\includegraphics[scale=0.18]{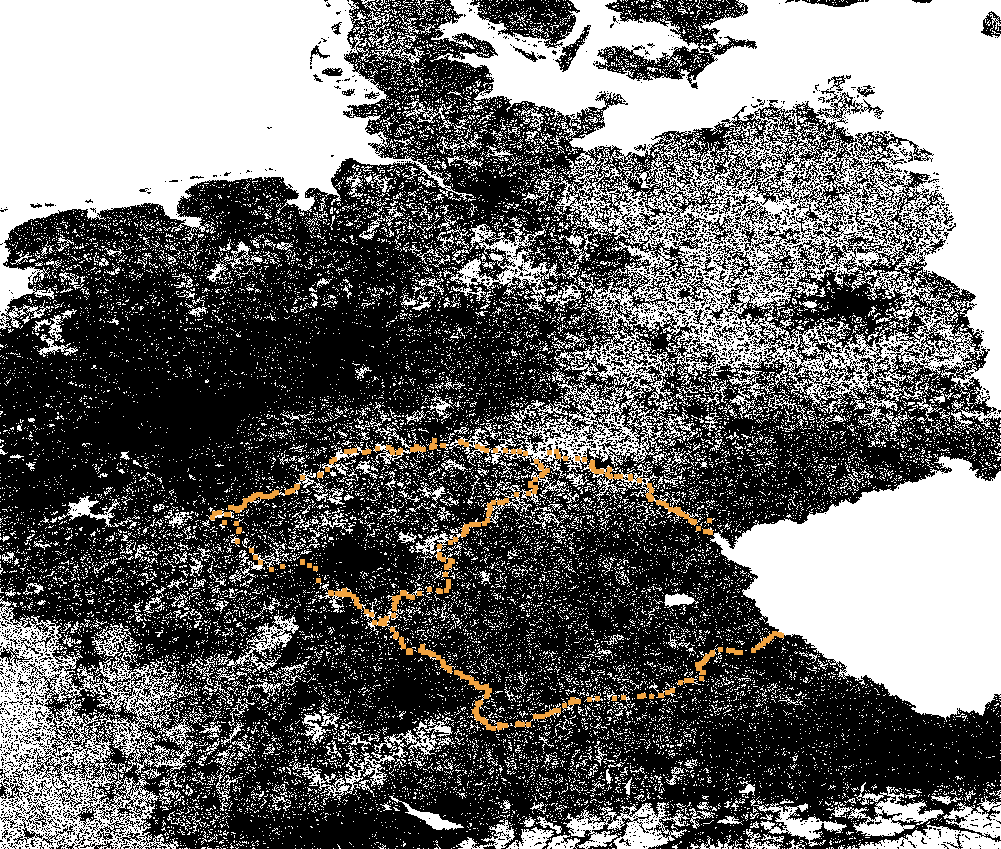}
\end{center}

\caption{Largest Bag in Road Graph Tree Decompositon.}
\label{fig:road-bag}
\end{wrapfigure}

Figure~\ref{fig:road-bag} depicts the nodes in a mostly German road graph at their geographical positions.
Every dot is a node.
The nodes in the largest bag of a tree decomposition are highlighted.
This bag has 369 nodes.
The Figure depicts an area surrounded by orange dots.
This area is divided into two parts along separator consisting of further orange dots.

The nodes in the area are the nodes in the cell.
The orange separator nodes are also in the cell.
However, the orange nodes that surrounded the area are not inside the cell.
They are in its boundary.

\subsection{Correctness}

\begin{lemma}
The constructed tree decomposition is valid.
\end{lemma}
\begin{proof}
We need to show that the three conditions laid out in the tree decomposition definition are fulfilled.
We need to show that every node is in a bag.
To prove this, we observe that every node is interior to the top level cell.
A node $v$ interior to a cell $c$ is either in $v$'s bag or interior to a child of $c$.
As every cell has a finite number of descendants, we cannot build infinite chains of nested cells.
We have thus proven that every node is in a bag.

Further, we need to show that for every edge $\{x,y\}$ there exists a bag $b$ such that $x$ and $y$ are part of $b$.
As touching cells are ordered by inclusion and as there are only finitely many cells, we know that there exists a smallest cell $c_x$ that has $x$ in its interior.
$x$ is in the bag of $c_x$ because $x$ is in $c_x$ but not in the interior of a child of $c_x$.
Let $c_y$ be the analogous smallest cell for $y$.
If $c_x = c_y$, then $c_x$ is the required $b$.
Otherwise, we observe that the existence of $\{x,y\}$ implies that $c_x$ and $c_y$ touch each other.
They are therefore ordered by inclusion.
Assume without loos of generality that $c_x \subseteq c_y$.
The existence of $\{x,y\}$ implies that $y$ is on the boundary of $c_x$ and therefore in the bag of $c_x$.
$c_x$ is therefore the required bag $b$.

Finally, we need to show that for every node $x$ the set of bags that include $x$ forms a subtree of the backbone.
Consider again the smallest cell $c_x$ that contains $x$.
All cells that contain $x$ are ancestors of $c_x$.
As $x$ is in $c_x$, $x$ cannot be in the bag of any ancestor of $c_x$.
We therefore know that $c_x$ is the only cell that has $x$ in its interior and bag.
Pick another cell $d$ whose bag contains $x$.
As $d \neq c_x$, we know that $x$ is on $d$'s boundary.
Denote by $p$ the parent cell of $d$.
$x$ is in $p$ or on the boundary of $p$.
We can thus conclude that $x$ is in the bag of $p$.
From $d$ we can iteratively follow the parent relation.
As the parent-child relation is acyclic, we eventually arrive at $c_x$.
As for every $d$ the corresponding path ends at $c_x$, we have proven that the set of bags that include $x$ forms a subtree of the backbone.

As we have proven all three properties, we have proven that the constructed tree decomposition is valid.
\end{proof}

\begin{lemma}
The constructed multilevel partition is valid.
\end{lemma}
\begin{proof}
We need to show that touching cells are ordered by inclusion. 
Cells touch if one of two conditions is fulfilled:
\begin{itemize}
\item They share a node.
\item There exists an edge with endpoints in both cells.
\end{itemize}
We show the required ordering property independently for both cases.

Consider some arbitrary node $x$ and denote by $T_x$ the subtree of the backbone induced by $x$.
$T_x$ contains a unique bag $b_x$ that is closest to the root.
The bags in the tree decomposition fall into three categories:
\begin{enumerate}
\item They are in $T_x$ but are different from $b_x$.
\item They lie on the unique path from $b_x$ to the root.
\item They are neither in $T_x$ nor on the path.
\end{enumerate}
We show that only the cells corresponding to the bags of category 2 contain $x$.
The cells along a path are trivially ordered by inclusion.
 
Let $q$ be a bag from category 1.
As it is different from $b_x$, it cannot be the root.
Therefore, there exists a parent bag $p$ of $q$.
By construction $x$ is also contained in $p$.
We have that $x\in q\cap p$.
$x$ is in the boundary and not in the interior of the cell corresponding to $q$.

Now let $q$ be a bag from category 3.
The corresponding cell is constructed by forming the union of bags that do not contain $x$.
The union thus also does not contain $x$.

Finally, let $q$ denote a bag on the path.
The constructed union contains all bags of $T_x$ and thus also contains $x$.
It remains to show that $x$ is not on the boundary of the corresponding cell.
This follows from the fact that $b_x$ is the only bag that contains $x$.
As $b_x$ is the bag the farthest away from the root, no parent bag of a bag on the path is $b_x$.
$x$ is thus not part of any boundary.

This completes the first part of the proof.
Next consider the case where the cells corresponding to two bags $b_x$ and $b_y$ touch because there exists an edge $\{x,y\}$ between them.
By convention, we set $x\in b_x$ and $y\in b_y$.
We know from the second property of the tree decomposition definition that there exists a bag $b_{xy}$ that contains $x$ and $y$.

We know that the tree backbone must contain the following four paths:
\begin{itemize}
\item There is a path $D_x$ from $b_{xy}$ to $b_x$ along $T_x$ because trees are connected. 
\item Using an analogous argument, we know that there exists a path $D_y$ from $b_{xy}$ to $b_y$ along $T_y$.
\item We can follow the parent relation from $b_x$ to the root and obtain a path $U_x$.
\item Analogously, there exists a path $U_y$ from $b_y$ to the root.
\end{itemize}
By concatenating all four paths $U_x$, $U_y$, $D_y$, and $D_x$, we obtain a cycle in the tree backbone.
As the tree backbone is a tree, we conclude that the cycle must be degenerate.
The set of the root $r$, $b_x$, $b_y$, and $b_{xy}$ can therefore only contain at most two elements.
We conclude that one of the following conditions must hold as otherwise the set would contain three or more elements:
$b_x = b_y$, $b_x = r$, or $b_y = r$.
In the first case, the cells are equal and thus ordered.
In the second and third cases, one of the cells is the root cell and thus by definition a superset of the other cell.

This completes the second and last part of the proof.
We have proven that the constructed multilevel partition is valid.
\end{proof}

\begin{lemma}
Applying both algorithms after another is the identity.
\end{lemma}
\begin{proof}
As both algorithms essentially just copy the tree backbone, it is clear that it is not modified.
It remains to show that the contents of the bags and cells remains unchanged.

Pick a cell $c$ from the multilevel partition and denote by $b$ the corresponding bag.
$b$ is by construction a subset of the union $c$'s boundary and its interior.
The missing nodes must be in the interior of a child cell of $c$ and therefore in the bag of a descendant of $b$.
The union over the bags of the subtree rooted at $b$ is thus equal to the union of $c$'s boundary and interior.
Removing the boundary from this union yields the interior, as the interior and the boundary are disjoint.
The interior of $c$ is thus left unchanged after applying both algorithms.
Applying both algorithms is therefore the identity.
\end{proof}

\section{PACE 2017 FlowCutter Submission}

Our algorithm is based on the nested dissection paradigm~\cite{g-ndrfe-73}.
In the graph partitioning literature, this approach is often called recursive bisection.
Compared to the original algorithms, ours does not recurse until all graph parts are smaller than a given threshold.
Instead, it aborts early, when it detects that the decomposition width can no longer be improved.

Internally, our algorithm uses the multilevel partition representation described in the previous section.
It translates the data to a tree decomposition when formatting the result.
For every cell $c$, it stores two sets $B_c$ and $I_c$.
Both sets contain nodes of the input graph.
$B_c$ contains the nodes on the boundary of $c$.
$I_c$ contains the interior nodes of $c$ that are not in a child cell.
The bag corresponding to $c$ is $B_c \cup I_c$.
We refer to $|B_c| + |I_c|$ as the bag size of $c$.
The maximum bag size over all cells in a multilevel partition is the corresponding tree decomposition width minus one.

Our algorithm maintains a multilevel partition.
The cells are organizes in two sets $\mathcal{F}$ and $\mathcal{O}$.
$\mathcal{F}$ is the set of final cells.
$\mathcal{O}$ is the set of open cells.
$\mathcal{F}$ contains the cells that can no longer be modified.
The cells in $\mathcal{O}$ can be modified using a cell split operation.
Splitting $c$ yields one final cell and several possibly zero open cells.
No new cell has a bag size larger than the one of $c$.

The set $\mathcal{O}$ is organized as maximum priority queue using the bag size as key.
Our algorithm further maintains the maximum bag size over all final cells.
Initially, the set $\mathcal{F}$ is empty.
We define the maximum bag size over an empty $\mathcal{F}$ as 0.
Initially,  $\mathcal{O}$ contains a single toplevel cell $c$ that contains all nodes.
This means that $B_c = \emptyset$ and $I_c = V$.

Our algorithm iteratively splits cells.
It picks the cell $c$ from $\mathcal{O}$ that maximizes the bag size.
If $c$'s bag size is smaller or equal to the maximum bag size over all final cells, our algorithm terminates.
Otherwise, it removes $c$ from $\mathcal{O}$ and splits $c$.
This yields a final cell $c_f$ and several possible zero open cells $c_1\ldots c_k$.
Our algorithm adds $c_f$ into $\mathcal{F}$ and updates the stored maximum bag size.
The cells $c_1\ldots c_k$ are added into $\mathcal{O}$.

\subsection{Splitting a Cell}

In this subsection, we describe how a cell $c$ is split.
Consider the subgraph $G_c$ induced by $I_c$.
Our algorithm computes a balanced separator $S_c$ using FlowCutter.
The resulting cells are derived from this separator.
The final cell $c_f$ is similar to $c$ except that nodes are removed from $I_c$.
Formally, $B_{c_f} = B_c$ and $I_{c_f} = S_c$.
The resulting open cells correspond to the connected components of $G_c\setminus S_c$.
For every component there is a cell $c_c$.
Our algorithm sets $I_{c_c}$ to the nodes in the corresponding component.
$B_{c_c}$ is a subset of $B_c \cup S_c$.
It only contains the nodes adjacent to a node in $I_{c_c}$.

To test the adjacency, we maintain an array that maps every node onto a bit.
Initially, we set each bit to false.
For all nodes in $I_{c_c}$, we set the bit to true.
Next we iterate over all nodes $x$ in $B_c \cup S_c$.
We iterate over the neighbors of $x$ in the input graph.
If the bit of one neighbor is set, we add $x$ to $B_{c_c}$.
Finally, we reset the bit of all nodes in $I_{c_c}$ to false.

\subsection{FlowCutter}

FlowCutter is a graph bisection algorithm described in~\cite{hs-gbpo-16}.
In this subsection, we present the high level ideas.
For all details, we refer to~\cite{hs-gbpo-16}.
The base algorithm computes balanced edge cuts.
It can be extended to compute balanced node separators.
Our algorithm interprets the graphs as symmetric unit flow network.
It repeatedly solves multi-source multi-target maximum flow problems.
Initially, there is only one source and one target node.
They are picked uniformly at random.
Our algorithm starts by computing a maximum flow.
From this flow, it derives the least balanced minimum cut $C$.
If $C$ is sufficiently balanced, it terminates.
Otherwise, it iteratively improves the cut in rounds.
New source or target nodes are added in each round.
Suppose that $C$ is unbalanced.
The algorithm thus continues with an additional round.
Further, assume without loose of generality that the source's side is smaller than the target's side.
Our algorithm marks all nodes on the source's side as additional sources.
Additionally, it marks one node on the target's side as source node.
This node is called \emph{piercing node}.
It is an endpoint of a cut edge.
Choosing the piercing node is difficult.
If there is a choice that does not increase the cut size, our algorithm picks it.
Otherwise, it employs a heuristic based on the distances from the original source and target nodes.
After adding the additional sources, $C$ no longer separates all source nodes.
In the next round a different cut is found.
Usually, we run several FlowCutter instances in parallel with different initial source and target pairs.
We refer to the number of instances as the number of \emph{cutters}.
In each round FlowCutter finds a new cut. 
During its execution our algorithm thus does not only compute a single balanced cut.
Instead, it computes a set of cuts that heuristically optimize balance and cut size in the Pareto-sense.
This allows us to optimize complex criteria heuristically.
For example, we can pick a cut that minimizes the ratio between the cut size and the smaller side size, i.e., with minimum expansion.
From our experience, minimizing the expansion subject to a bounded balance, heuristically yields good results.
In flow problems, node capacities can usually be reduced to edge capacities~\cite{amo-nf-93}.
Using such a reduction, FlowCutter can be used to compute node separators.
Alternatively, node separators can be derived from edge cuts by picking for every cut edge an endpoint.
The former can achieve smaller separators but is slower as the graph needs to be expanded.

\subsection{Submission Details}

Our submission runs FlowCutter in an endless loop until the maximum execution time of 30 min is reached.
It uses node capacities and minimizes the expansion subject to a bounded balance.
In each iteration, our submission varies some parameters such as the number of cutters, the minimum required balance, or the random seed.
For small instances, we compute heuristic elimination orders before running FlowCutter.
From these orders, we derive tree decompositions.
We compute an order that greedily picks a minimum degree node.
Another order greedily picks a minimum fill-in node.
We compute these orders because the derived tree decompositions are sometimes slightly smaller than those computed by FlowCutter.
However, on large instances computing them using our implementation is not possible within the maximum execution time of 30 min.
On the large instances, we configure FlowCutter to compute edge cuts and minimize the cut size subject to a bounded balance.
From the edge cuts, we derive node separators by picking an endpoint of every edge.
This is necessary to assure that at least one iteration finishes within the allowed execution time.
Our implementation finds a new tree decomposition each time a cell is split.
Formatting the textual output corresponding to the new tree decomposition can be slower than performing the split operation.
To avoid these costs, we produce a new textual output only every 30 s.

\section{Evaluation}

\begin{table}
\begin{center}
\begin{tabular}{lrrrrrr}
\toprule
&\multicolumn{2}{c}{Approx. Ratio}  & \multicolumn{4}{c}{Diff. to best Sol. [\%]}  \\
\cmidrule(lr){2-3}\cmidrule(lr){4-7}
Competitor & Avg. & Max. & $=0$  & $\le1$ & $\le2$ & $<\infty$ \\
\midrule
FlowCutter & \textbf{1.08} & \textbf{1.55} & 42 & 50 & \textbf{56} & \textbf{100} \\
Solver by Abseher, Musliu, Woltran & 1.11 & 2.37 & 25 & 37 & 46 & 95 \\
Solver by Bannach, Berndt, Ehlers & 1.19 & 3.90 & 21 & 29 & 32 & 93 \\
Solver 1 by Jégou, Kanso, Terrioux & 1.30 & 2.43 & 9 & 13 & 17 & 81 \\
Solver by Tamaki, Ohtsuka, Sato, Makii & 1.30 & 9.61 & \textbf{51} & \textbf{52} & 52 & 92 \\
Solver by Larisch, Salfelder & 1.56 & 3.74 & 6 & 8 & 13 & 54 \\
Solver 2 by Jégou, Kanso, Terrioux & 4.68 & 139.70 & 10 & 14 & 18 & 71 \\
\bottomrule
\end{tabular}

\end{center}

\caption{Performance Comparison of all PACE 2017 Track A2 competitors.}
\label{tab:result}
\end{table}

To evaluate our algorithm, we submitted it the PACE 2017 challenge.
We entered track A2.
The objective is to compute a tree decomposition of small size within 30 min.
All competitors were evaluated on the same set of test graphs.
The algorithms that find a smallest decomposition win the instance.
Entries were ranked using the Schulze method.
This usually implies that the competitor that wins the most instances wins the competition. 
Our submission won the second place.
The submission by Tamaki et al. won the competition.
Abseher et al. is ranked third.
The difference between the first two places is small.
In the following, we look at the results in depth.

Denote by $t^*$ the best solution found for an instance and by $t$ the solution found by a competitor.
A weakness of the PACE evaluation method is that is does not consider how large $t-t^*$ is.
$t-t^* = 1$ is weighted in the same way as $t-t^* = 1000$.
Only the relative ranking among competitors matters.
To investigate the impact of this effect, we present in Table~\ref{tab:result} alternative rankings.
We report relative approximation ratios.
Denote by $t^*$ the best solution found for an instance and by $t$ the solution found by a competitor.
We report the average and maximum $t/t^*$ over all non-trivially\footnote{A non-trivial solution must have width at most $|V|-5$. Criterion was defined by PACE organizers.} solved instances.
With respect to both criteria, our submission clearly wins.
Interestingly, the competition winner is only ranked fifth with respect to the average $t/t^*$.
Additionally, we report how often $t-t^*$ is below a threshold.
For instances that were not or only trivially solved, we set $t=\infty$.
``$=0$'' is the number of instances where a competitor produced a best solution.
``$\le X$'' indicates how often it was off by $X$.
``$< \infty$'' is the number of non-trivially solved instances.  

The algorithm by Tamaki et al. clearly wins with respect to ``$=0$''.
The gap with our algorithm is only 2\% for ``$\le 1$''.
For ``$\le 2$'' our algorithm wins.
Further, our algorithm is the only competitor that finds non-tririval solutions on all instances.
The algorithm by Tamaki et al. finds on many small instances solutions that are one or two nodes smaller than our algorithm.
For large instances, their algorithm often does not find a solution.

\section{Conclusion}

We conclude that the winning algorithm by Tamaki et al. is good, if the instances are small and achieving the smallest width is of utmost importance.
If the instances are large or finding nearly the best solutions is good enough, our algorithm is better.
Further, we illustrated that viewing tree decompositions as multilevel partitions is a useful algorithm design tool.

\subparagraph*{Acknowledgements.}

I thank the PACE organizers for providing a good testing infrastructure.
It made the experimental evaluation easier and less error-prone.
I thank Michael Hamann for fruitful discussions.




\begin{thebibliography}{10}

\bibitem{acdgw-o-16}
Ittai Abraham, Shiri Chechik, Daniel Delling, Andrew~V. Goldberg, and Renato~F.
  Werneck.
\newblock On dynamic approximate shortest paths for planar graphs with
  worst-case costs.
\newblock In SODA'16, pages 740--753. SIAM, 2016.

\bibitem{amo-nf-93}
Ravindra~K. Ahuja, Thomas~L. Magnanti, and James~B. Orlin.
\newblock {\em Network Flows: Theory, Algorithms, and Applications}.
\newblock Prentice Hall, 1993.

\bibitem{bdgmpsww-rptn-16}
Hannah Bast, Daniel Delling, Andrew~V. Goldberg, Matthias
  {M{\"u}ller--Hannemann}, Thomas Pajor, Peter Sanders, Dorothea Wagner, and
  Renato~F. Werneck.
\newblock Route planning in transportation networks.
\newblock In {\em Algorithm Engineering - Selected Results and Surveys}, volume
  9220 of LNCS, pages 19--80. Springer,
  2016.

\bibitem{bcrw-s-16}
Reinhard Bauer, Tobias Columbus, Ignaz Rutter, and Dorothea Wagner.
\newblock Search-space size in contraction hierarchies.
\newblock {\em Theoretical Computer Science}, 645:112--127, 2016.

\bibitem{bp-aicgc-93}
Jean Blair and Barry Peyton.
\newblock An introduction to chordal graphs and clique trees.
\newblock In {\em Graph Theory and Sparse Matrix Computation}, volume~56 of
  {\em The IMA Volumes in Mathematics and its Applications}, pages 1--29.
  Springer, 1993.

\bibitem{b-atgt-93}
Hans~L. Bodlaender.
\newblock A tourist guide through treewidth.
\newblock {\em Acta Cybernetica}, 11:1--21, 1993.

\bibitem{b-tsa-07}
Hans~L. Bodlaender.
\newblock Treewidth: Structure and algorithms.
\newblock In {\em Proceedings of the 14th International Colloquium on
  Structural Information and Communication Complexity}, volume 4474 of LNCS, pages 11--25. Springer, 2007.

\bibitem{bmsss-ragp-13}
Aid{\i}n Bulu{\c{c}}, Henning Meyerhenke, Ilya Safro, Peter Sanders, and
  Christian Schulz.
\newblock Recent advances in graph partitioning, 2013.
\newblock arXiv:1311.3144 [cs.DS].

\bibitem{dhjkkr-tfpac-16}
Holger Dell, Thore Husfeldt, Bart~M. Jansen, Petteri Kaski, Christian
  Komusiewicz, and Frances Rosamond.
\newblock The first parameterized algorithms and computational experiments
  challenge.
\newblock In IPEC'16, pages
  30:1--30:9, 2016.

\bibitem{dgpw-crprn-13}
Daniel Delling, Andrew~V. Goldberg, Thomas Pajor, and Renato~F. Werneck.
\newblock Customizable route planning in road networks.
\newblock {\em Transportation Science}, 51(2):566--591, 2017.

\bibitem{bingcrp}
Bing Developers.
\newblock Bing maps new routing engine.
\newblock Website, 2012.
\newblock Online at
  \url{https://blogs.bing.com/maps/2012/01/05/bing-maps-new-routing-engine}.

\bibitem{dsw-cch-15}
Julian Dibbelt, Ben Strasser, and Dorothea Wagner.
\newblock Customizable contraction hierarchies.
\newblock {\em ACM Journal of Experimental Algorithmics}, 21(1):1.5:1--1.5:49,
  April 2016.

\bibitem{fg-imig-65}
Delbert~R. Fulkerson and O.~A. Gross.
\newblock Incidence matrices and interval graphs.
\newblock {\em Pacific Journal of Mathematics}, 15(3):835--855, 1965.

\bibitem{gssv-erlrn-12}
Robert Geisberger, Peter Sanders, Dominik Schultes, and Christian Vetter.
\newblock Exact routing in large road networks using contraction hierarchies.
\newblock {\em Transportation Science}, 46(3), 2012.

\bibitem{g-ndrfe-73}
Alan George.
\newblock Nested dissection of a regular finite element mesh.
\newblock {\em SIAM Journal on Numerical Analysis}, 10(2):345--363, 1973.

\bibitem{hs-gbpo-16}
Michael Hamann and Ben Strasser.
\newblock Graph bisection with pareto-optimization.
\newblock In ALENEX'16, pages 90--102. SIAM, 2016.

\bibitem{hsw-emlog-08}
Martin Holzer, Frank Schulz, and Dorothea Wagner.
\newblock Engineering multilevel overlay graphs for shortest-path queries.
\newblock {\em ACM JEA}, 13(2.5):1--26,
  December 2008.

\bibitem{sww-daola-00}
Frank Schulz, Dorothea Wagner, and Karsten Weihe.
\newblock {D}ijkstra's algorithm on-line: An empirical case study from public
  railroad transport.
\newblock {\em ACM JEA}, 5(12):1--23, 2000.

\end{thebibliography}


\end{document}